\documentclass[10pt,nocopyrightspace]{sigplanconf}
\pdfoutput=1
\usepackage{listings}
\usepackage{enumitem}
\usepackage[hyphens]{url}
\usepackage[svgnames]{xcolor}
\definecolor{lcl}{RGB}{140,0,100}
\usepackage[colorlinks=true,breaklinks,draft=false]{hyperref}
\hypersetup{urlcolor=lcl,linkcolor=lcl,citecolor=lcl}
\newcommand{\doi}[1]{doi:~\href{http://dx.doi.org/#1}{\Hurl{#1}}}
\usepackage{amssymb}
\usepackage{multicol}
\usepackage{xcolor}
\usepackage{graphicx}
\usepackage[fleqn]{amsmath}
\usepackage{graphics}
\usepackage{stmaryrd}
\usepackage{amsthm}
\usepackage{ftnright}
\usepackage[T1]{fontenc}
\usepackage{semantic}
\usepackage{enumitem}
\usepackage[hang,flushmargin]{footmisc}
\usepackage[notquote]{hanging}
\usepackage{flushend}
\usepackage{etoolbox}

\makeatletter
\renewcommand{\@makefntext}[1]{%
  \parindent 1em%
  \raggedright
  \begin{hangparas}{0.8em}{1}
  \noindent {$^{\@thefnmark}$~#1}
  \end{hangparas}
}
\makeatother

\newcommand{\langl}{\begin{picture}(4.5,7)
\put(1.1,2.5){\rotatebox{60}{\line(1,0){5.5}}}
\put(1.1,2.5){\rotatebox{300}{\line(1,0){5.5}}}
\end{picture}}
\newcommand{\rangl}{\begin{picture}(4.5,7)
\put(.9,2.5){\rotatebox{120}{\line(1,0){5.5}}}
\put(.9,2.5){\rotatebox{240}{\line(1,0){5.5}}}
\end{picture}}

\definecolor{cmtclr}{rgb}{0.0,0.6,0.0}
\definecolor{numclr}{rgb}{0.0,0.4,0.0}
\definecolor{kvdclr}{rgb}{0.0,0.0,0.6}
\definecolor{strclr}{rgb}{0.5,0.1,0.0}
\definecolor{prepclr}{rgb}{0.0,0.0,0.0}

\newcommand{\kvd}[1]{\textnormal{\textcolor{kvdclr}{\sffamily #1}}}
\newcommand{\num}[1]{\textnormal{\textcolor{numclr}{\sffamily #1}}}
\newcommand{\str}[1]{\textnormal{\textcolor{strclr}{\sffamily "#1"}}}
\newcommand{\strf}[1]{\textnormal{\textcolor{strclr}{\sffamily #1}}}
\newcommand{\ident}[1]{\textnormal{\sffamily #1}}

\newcommand{\lsep}[0]{\;\; | \;\;}
\newcommand{\narrow}[1]{\hspace{-0.7em} #1 \hspace{-0.7em}}

\newcommand{\tytag}{\ident{tag}}
\newcommand{\dropopt}[1]{\lfloor#1\rfloor}
\newcommand{\addopt}[1]{\lceil#1\rceil}
\newcommand{\tytagof}{\ident{tagof}}

\newcommand{\reduce}{\rightsquigarrow}

\newcommand{\sem}[1]{\llbracket #1 \rrbracket}
\newcommand{\semalt}[1]{S(#1)}

\newtheorem{definition}{Definition}
\newtheorem{theorem}{Theorem}
\newtheorem{remark}{Remark}

\newtheorem{lemma}[theorem]{Lemma}

\begin{document}

\setlength{\pdfpageheight}{\paperheight}
\setlength{\pdfpagewidth}{\paperwidth}
\conferenceinfo{PLDI '16}{June 13--17, 2016, Santa Barbara, CA, United States}
\copyrightyear{2016}
\copyrightdata{978-1-nnnn-nnnn-n/yy/mm}


\title{Types from data: \textnormal{Making structured data first-class citizens in F\#}}

\authorinfo{Tomas Petricek}
           {University of Cambridge}
           {tomas@tomasp.net}
\authorinfo{Gustavo Guerra}
           {Microsoft Corporation, London}
           {gustavo@codebeside.org}
\authorinfo{Don Syme}
           {Microsoft Research, Cambridge}
           {dsyme@microsoft.com}
\maketitle


\begin{abstract}
Most modern applications interact with external services and access data in structured formats such
as XML, JSON and CSV. Static type systems do not understand such formats, often making data access
more cumbersome. Should we give up and leave the messy world of external data to dynamic typing
and runtime checks? Of course, not!

We present F\# Data, a library that integrates external structured data into F\#. As most real-world
data does not come with an explicit schema, we develop a shape inference algorithm that infers a
shape from representative sample documents. We then integrate the inferred shape into the F\# type
system using type providers. We formalize the process and prove a relative type soundness theorem.

Our library significantly reduces the amount of data access code and it provides additional
safety guarantees when contrasted with the widely used weakly typed techniques.
\end{abstract}


%
%

\section{Introduction}
\label{sec:introduction}

Applications for social networks, finding tomorrow's weather or searching train schedules
all communicate with external services. Increasingly, these services provide end-points that return
data as CSV, XML or JSON. Most such services do not come with an explicit schema. At best, the
documentation provides sample responses for typical requests.

For example, \url{http://openweathermap.org/current} contains one example to document an end-point
to get the current weather. Using standard libraries, we might call it as\footnote{We abbreviate
the full URL and omit application key (available after registration). The returned JSON is shown
in Appendix~\ref{sec:appendix-weather} and can be used to run the code against a local file.}:

\begin{equation*}
\begin{array}{l}
 \kvd{let}~\ident{doc}=\ident{Http.Request}(\str{http://api.owm.org/?q=NYC}) \\
 \kvd{match}~\ident{JsonValue.Parse}(\ident{doc})~\kvd{with} \\
 |~\ident{Record}(\ident{root})\rightarrow \\
 \quad \kvd{match}~\ident{Map.find}~\str{main}~\ident{root}~\kvd{with} \\
 \quad |~\ident{Record}(\ident{main})\rightarrow \\
 \quad \quad \kvd{match}~\ident{Map.find}~\str{temp}~\ident{main}~\kvd{with} \\
 \quad \quad |~\ident{Number}(\ident{num})\rightarrow \ident{printfn}~\str{Lovely \%f!}~\ident{num} \\
 \quad \quad |~\_\rightarrow \ident{failwith}~\str{Incorrect format} \\
 \quad |~\_\rightarrow \ident{failwith}~\str{Incorrect format} \\
 |~\_\rightarrow \ident{failwith}~\str{Incorrect format}
\end{array}
\end{equation*}
The code assumes that the response has a particular shape described in the documentation. The
root node must be a record with a \strf{main} field, which has to be another record containing
a numerical \strf{temp} field representing the current temperature. When the shape is different,
the code fails. While not immediately unsound, the code is prone to errors if strings are
misspelled or incorrect shape assumed.

Using the JSON type provider from F\# Data, we can write code with exactly the
same functionality in two lines:
\vspace{-0.1em}
\begin{equation*}
\begin{array}{l}
 \kvd{type}~\ident{W} = \ident{JsonProvider}\langl\str{http://api.owm.org/?q=NYC}\rangl \\[0.1em]
 \ident{printfn}~\str{Lovely \%f!}~(\ident{W.GetSample().Main.Temp})
\end{array}
\end{equation*}
$\ident{JsonProvider}\langl\str{...}\rangl$ invokes a type provider \cite{fsharp-typeprov} at
compile-time with the URL as a sample. The type provider infers the structure of the response
and provides a type with a \ident{GetSample} method that returns a parsed JSON with nested
properties \ident{Main.Temp}, returning the temperature as a number.

In short, the \emph{types} come from the sample \emph{data}. In our experience, this technique is
both practical and surprisingly effective in achieving more sound information interchange
in heterogeneous systems. Our contributions are as follows:

\begin{itemize}
\item We present F\# Data type providers for XML, CSV and JSON (\S\ref{sec:providers})
  and practical aspects of their implementation that contributed to their industrial
  adoption (\S\ref{sec:impl}).

\item We describe a predictable shape inference algorithm for structured data formats,
  based on a \emph{preferred shape} relation, that underlies the type providers
  (\S\ref{sec:inference}).

\item We give a formal model (\S\ref{sec:formal}) and use it to prove
  \emph{relative type safety} for the type providers (\S\ref{sec:safety}).
\end{itemize}

%
%

\section{Type providers for structured data}
\label{sec:providers}

We start with an informal overview that shows how F\# Data type providers simplify working with
JSON and XML. We introduce the necessary aspects of F\# type providers along the way. The examples
in this section also illustrate the key design principles of the shape inference algorithm:

\begin{itemize}
\item The mechanism is predictable (\S\ref{sec:impl-stable}). The user directly works with the
  provided types and should understand why a specific type was produced from a given sample.

\item The type providers prefer F\# object types with properties. This allows extensible
  (open-world) data formats (\S\ref{sec:providers-xml}) and it interacts well with developer tooling
  (\S\ref{sec:providers-json}).

\item The above makes our techniques applicable to any language with nominal
  object types (e.g.~variations of Java or C\# with a type provider mechanism added).

\item Finally, we handle practical concerns including
  support for different numerical types, \kvd{null} and missing data.
\end{itemize}

\noindent
The supplementary screencast provides further illustration of the practical developer
experience using F\# Data.\footnote{Available at \url{http://tomasp.net/academic/papers/fsharp-data}.}


\subsection{Working with JSON documents}
\label{sec:providers-json}

The JSON format is a popular data exchange format based on
JavaScript data structures. The following is the definition of \ident{JsonValue}
used earlier (\S\ref{sec:introduction}) to represent JSON data:
\begin{equation*}
\begin{array}{l}
 \kvd{type}~\ident{JsonValue} = \\[0.1em]
 \quad|~ \ident{Number}~\kvd{of}~\ident{float} \\[0.1em]
 \quad|~ \ident{Boolean}~\kvd{of}~\ident{bool} \\[0.1em]
 \quad|~ \ident{String}~\kvd{of}~\ident{string}\\[0.1em]
 \quad|~ \ident{Record}~\kvd{of}~\ident{Map}\langl\ident{string}, \ident{JsonValue}\rangl \\[0.1em]
 \quad|~ \ident{Array}~\kvd{of}~\ident{JsonValue}[] \\[0.1em]
 \quad|~ \ident{Null} \\[0.1em]
\end{array}
\end{equation*}
The earlier example used only a nested record containing a number. To demonstrate other
aspects of the JSON type provider, we look at an example that also involves an array:
{\small{
\begin{verbatim}
  [ { "name":"Jan", "age":25 },
    { "name":"Tomas" },
    { "name":"Alexander", "age":3.5 } ]
\end{verbatim}
}}
\noindent
The standard way to print the names and ages would be to pattern match on the parsed
\ident{JsonValue}, check that the top-level node is a \ident{Array} and iterate over the elements
checking that each element is a \ident{Record} with certain properties. We would throw an exception
for values of an incorrect shape. As before, the code would specify field names as strings, which
is error prone and can not be statically checked.

Assuming \strf{people.json} is the above example and \ident{data} is a string containing
JSON of the same shape, we can write:
\begin{equation*}
\begin{array}{l}
 \kvd{type}~\ident{People}~=~\ident{JsonProvider}\langl\str{people.json}\rangl\hspace{1em} \\[0.6em]
 \kvd{for}~\ident{item}~\kvd{in}~\ident{People.Parse}(\ident{data})~\kvd{do}\\[0.1em]
 \quad\ident{printf}~\str{\%s }~\ident{item.Name}\\[0.1em]
 \quad\ident{Option.iter}~(\ident{printf}~\str{(\%f)})~\ident{item.Age}
\end{array}
\end{equation*}
We now use a local file as a sample for
the type inference, but then processes data from another source. The code achieves a similar
simplicity as when using dynamically typed languages, but it is statically type-checked.

\paragraph{Type providers.}
The notation $\ident{JsonProvider}\langl\str{people.json}\rangl$ passes a \emph{static parameter}
to the type provider. Static parameters are resolved at compile-time and have to be constant.
The provider analyzes the sample and provides a type  \ident{People}. F\# editors also
execute the type provider at development-time and use the provided types
for auto-completion on ``.'' and for background type-checking.

The \ident{JsonProvider} uses a shape inference algorithm and provides
the following F\# types for the sample:
\begin{equation*}
\begin{array}{l}
 \kvd{type}~\ident{Entity}~=  \\[-0.05em]
 \quad \kvd{member}~\ident{Name}~:~\ident{string} \\[-0.05em]
 \quad \kvd{member}~\ident{Age}~:~\ident{option}\langl \ident{float}\rangl \\[0.4em]
 \kvd{type}~\ident{People}~=  \\[-0.05em]
 \quad \kvd{member}~\ident{GetSample}~:~\ident{unit}~\rightarrow~\ident{Entity}[] \\[-0.05em]
 \quad \kvd{member}~\ident{Parse}~:~\ident{string}~\rightarrow~\ident{Entity}[] \\[-0.05em]
\end{array}
\end{equation*}
The type \ident{Entity} represents the person. The field \ident{Name} is available for all
sample values and is inferred as \ident{string}. The field \ident{Age} is marked as optional,
because the value is missing in one sample. In F\#, we use \ident{Option.iter} to call
the specified function (printing) only when an optional value is available. The two age values
are an integer $25$ and a float $3.5$ and so the common inferred type is \ident{float}.
The names of the properties are normalized to follow standard F\# naming conventions
as discussed later (\S\ref{sec:impl-naming}).

The type \ident{People} has two methods for reading data. \ident{GetSample} parses the
sample used for the inference and \ident{Parse} parses a JSON string. This lets us read
data at runtime, provided that it has the same shape as the static sample.

\paragraph{Error handling.}
In addition to the structure of the types, the type provider also specifies the code of
operations such as \ident{item.Name}. The runtime behaviour is
the same as in the earlier hand-written sample (\S\ref{sec:introduction}) -- a member access
throws an exception if data does not have the expected shape.

Informally, the safety property (\S\ref{sec:safety}) states that if the inputs are compatible
with one of the static samples (i.e.~the samples are representative), then no exceptions will
occur. In other words, we cannot avoid all failures, but we can prevent some. Moreover, if
\url{http://openweathermap.org} changes the shape of the response, the code in \S\ref{sec:introduction}
will not re-compile and the developer knows that the code needs to be corrected.

\paragraph{Objects with properties.}
The sample code is easy to write thanks to the fact that most F\# editors provide auto-completion
when ``.'' is typed (see the supplementary screencast). The developer does not need to examine the
sample JSON file to see what fields are available. To support this scenario, our type providers
map the inferred shapes to F\# objects with (possibly optional) properties.

This is demonstrated by the fact that \ident{Age} becomes an optional member.
An alternative is to provide two different record types (one with \ident{Name} and one with
\ident{Name} and \ident{Age}), but this would complicate the processing code.
It is worth noting that languages with stronger tooling around pattern matching
such as Idris \cite{idris-tools} might have different preferences.


\subsection{Processing XML documents}
\label{sec:providers-xml}

XML documents are formed by nested elements with attributes. We can view elements as records with
a field for each attribute and an additional special field for the nested contents (which is a
collection of elements).

Consider a simple extensible document format where a root element {\ttfamily\small <doc>} can
contain a number of document elements, one of which is {\ttfamily\small <heading>} representing
headings:
{\small{
\begin{verbatim}
  <doc>
    <heading>Working with JSON</heading>
    <p>Type providers make this easy.</p>
    <heading>Working with XML</heading>
    <p>Processing XML is as easy as JSON.</p>
    <image source="xml.png" />
  </doc>
\end{verbatim}
}}
\noindent
The F\# Data library has been designed primarily to simplify reading of data. For example,
say we want to print all headings in the document. The sample shows a part of the document structure
(in particular the {\ttfamily\small <heading>} element), but it does not show all possible elements
(say, {\ttfamily\small <table>}). Assuming the above document is \strf{sample.xml}, we can write:
\noindent
\begin{equation*}
\begin{array}{l}
 \kvd{type}~\ident{Document}~=~\ident{XmlProvider}\langl\str{sample.xml}\rangl\hspace{1em} \\[0.5em]
 \kvd{let}~\ident{root}~=~\ident{Document.Load}(\str{pldi/another.xml})\\
 \kvd{for}~\ident{elem}~\kvd{in}~\ident{root.Doc}~\kvd{do}\\
 \quad\ident{Option.iter}~(\ident{printf}~\str{ - \%s})~\ident{elem.Heading}\\
\end{array}
\end{equation*}
The example iterates over a collection of elements returned by \ident{root.Doc}. The type of \ident{elem}
provides typed access to elements known statically from the sample and so we can write \ident{elem.Heading},
which returns an optional string value.

\paragraph{Open world.}
By its nature, XML is extensible and the sample cannot include all possible nodes.\footnote{Even
when the document structure is defined using XML Schema, documents may contain elements prefixed
with other namespaces.} This is the fundamental \emph{open world assumption} about external data.
Actual input might be an element about which nothing is known.

For this reason, we do not infer a closed choice between heading, paragraph and image. In the
subsequent formalization, we introduce a \emph{top shape} (\S\ref{sec:inference-types}) and extend
it with labels capturing the statically known possibilities (\S\ref{sec:inference-vars}). The
\emph{labelled top shape} is mapped to the following type:
\begin{equation*}
\begin{array}{l}
 \kvd{type}~\ident{Element}~=  \\[0.1em]
 \quad \kvd{member}~\ident{Heading}~:~\ident{option}\langl \ident{string} \rangl\\
 \quad \kvd{member}~\ident{Paragraph}~:~\ident{option}\langl \ident{string} \rangl\\
 \quad \kvd{member}~\ident{Image}~:~\ident{option}\langl \ident{Image} \rangl\\
\end{array}
\end{equation*}
\ident{Element} is an abstract type with properties. It can represent the statically known
elements, but it is not limited to them. For a table element, all three properties would return
\ident{None}.

Using a type with optional properties provides access to the elements known statically from the
sample. However the user needs to explicitly handle the case when a value is not a statically
known element. In object-oriented languages, the same could be done by providing a class hierarchy,
but this loses the easy discoverability when ``.'' is typed.

The provided type is also consistent with our design principles, which prefers optional properties.
The gain is that the provided types support both open-world data and developer tooling. It is also
worth noting that our shape inference uses labelled top shapes only as the last resort
(Lemma~\ref{thm:lub}, \S\ref{sec:impl-hetero}).


\subsection{Real-world JSON services}
\label{sec:providers-sum}

Throughout the introduction, we used data sets that demonstrate the typical problems frequent
in the real-world (missing data, inconsistent encoding of primitive values and heterogeneous shapes).
The government debt information returned by the World Bank\footnote{Available at
\url{http://data.worldbank.org}} includes all three:
{\small{
\begin{verbatim}
  [ { "pages": 5 },
    [ { "indicator": "GC.DOD.TOTL.GD.ZS",
        "date": "2012", "value": null },
      { "indicator": "GC.DOD.TOTL.GD.ZS",
        "date": "2010", "value": "35.14229" } ] ]
\end{verbatim}
}}
\noindent
First, the field \ident{value} is \kvd{null} for some records. Second, numbers in JSON can be
represented as numeric literals (without quotes), but here, they are returned as string literals
instead.\footnote{This is often used to avoid non-standard numerical types of JavaScript.}
Finally, the top-level element is a collection containing two values of different shape.
The record contains meta-data with the total number of pages and the array contains the data.
F\# Data supports a concept of heterogeneous collection (outlined in in \S\ref{sec:impl-hetero})
and provides the following type:
\begin{equation*}
\begin{array}{l}
\kvd{type}~\ident{Record}~=  \\
 \quad \kvd{member}~\ident{Pages}~:~\ident{int} \\[0.5em]
\kvd{type}~\ident{Item}~=  \\
 \quad \kvd{member}~\ident{Date}~:~\ident{int} \\
 \quad \kvd{member}~\ident{Indicator}~:~\ident{string} \\
 \quad \kvd{member}~\ident{Value}~:~\ident{option}\langl \ident{float}\rangl\\[0.5em]
\kvd{type}~\ident{WorldBank}~=  \\
 \quad \kvd{member}~\ident{Record}~:~\ident{Record}\\
 \quad \kvd{member}~\ident{Array}~:~\ident{Item}[]\\[0.5em]
\end{array}
\end{equation*}
The inference for heterogeneous collections infers the multiplicities and shapes of nested
elements. As there is exactly one record and one array, the provided type \ident{WorldBank} exposes
them as properties \ident{Record} and \ident{Array}.

In addition to type providers for JSON and XML, F\# Data also implements a type provider for CSV
(\S\ref{sec:impl-parsing}). We treat CSV files as lists of records (with field for each column)
and so CSV is handled directly by our inference algorithm.

%
%

\section{Shape inference for structured data}
\label{sec:inference}

The shape inference algorithm for structured data is based on a shape preference relation. When
inferring the shape, it infers the most specific shapes of individual values (CSV rows, JSON or XML
nodes) and recursively finds a common shape of all child nodes or all sample documents.

We first define the shape of structured data $\sigma$. We use the term \emph{shape} to distinguish
shapes of data from programming language \emph{types} $\tau$ (type providers generate the latter from the former).
Next, we define the preference relation on shapes $\sigma$ and describe the algorithm
for finding a common shape.

The shape algebra and inference presented here is influenced by the design principles
we outlined earlier and by the type definitions available in the F\# language.
The same principles apply to other languages, but details may differ, for example
with respect to numerical types and missing data.


\subsection{Inferred shapes}
\label{sec:inference-types}

We distinguish between \emph{non-nullable shapes} that always have a valid value (written as
$\hat{\sigma}$) and \emph{nullable shapes} that encompass missing and \kvd{null} values
(written as $\sigma$). We write $\nu$ for record names and record field names.
\begin{equation*}
\begin{array}{rcl}
 \hat{\sigma} &\narrow{=}& \nu \; \{ \nu_1 \!:\! \sigma_1, \ldots, \nu_n \!:\! \sigma_n,\; \rho_i  \} \\[0.1em]
                &\narrow{|}& \ident{float} \lsep \ident{int} \lsep \ident{bool} \lsep \ident{string}
 \\[0.6em]
       \sigma &\narrow{=}& ~\hat{\sigma}~ \lsep \kvd{nullable}\langl \hat{\sigma} \rangl \lsep [\sigma] \lsep \kvd{any} \lsep \kvd{null}  \lsep ~\bot~
\end{array}
\end{equation*}

\noindent
Non-nullable shapes include records (consisting of a name and fields with their shapes) and
primitives. The row variables $\rho_i$ are discussed below. Names of records arising from XML are the names of the XML elements.
For JSON records we always use a single name $\,\bullet$. We assume that record fields can be freely
reordered.

We include two numerical primitives, \ident{int} for integers and \ident{float} for floating-point
numbers. The two are related by the preference relation and we prefer \ident{int}.

Any non-nullable shape $\hat{\sigma}$ can be wrapped as $\kvd{nullable}\langl\hat{\sigma}\rangl$ to
explicitly permit the \kvd{null} value. Type providers map \kvd{nullable} shapes to the F\# option
type. A collection $[\sigma]$ is also nullable and \kvd{null} values are treated as empty
collections. This is motivated by the fact that a \kvd{null} collection is usually
handled as an empty collection by client code. However there is a range of design alternatives (make collections
non-nullable or treat \kvd{null} \ident{string} as an empty string).

The shape $\kvd{null}$ is inhabited by the $\kvd{null}$ value (using an overloaded
notation) and $\bot$ is the bottom shape. The \kvd{any} shape is the top shape, but we later add
labels for statically known alternative shapes (\S\ref{sec:inference-vars}) as
discussed earlier (\S\ref{sec:providers-xml}).

During inference we use row-variables $\rho_i$ \cite{rows-cardelli} in record shapes to represent
the flexibility arising from records in samples. For example, when a
record $\ident{Point}\,\{ \ident{x} \mapsto \num{3} \}$ occurs in a sample,
it may be combined with $\ident{Point}\,\{\, \ident{x}\mapsto\num{3}, \ident{y}\mapsto\num{4} \}$ that contains more fields. The
overall shape inferred must account for the fact that any extra fields are optional,
giving an inferred shape $\ident{Point}\,\{ \ident{x}\!:\!\ident{int},~ \ident{y}\!:\!\kvd{nullable}\langl\ident{int}\rangl\}$.


\begin{figure}
\begin{center}
\includegraphics[scale=0.80,trim=5mm 5mm 5mm 5mm,clip]{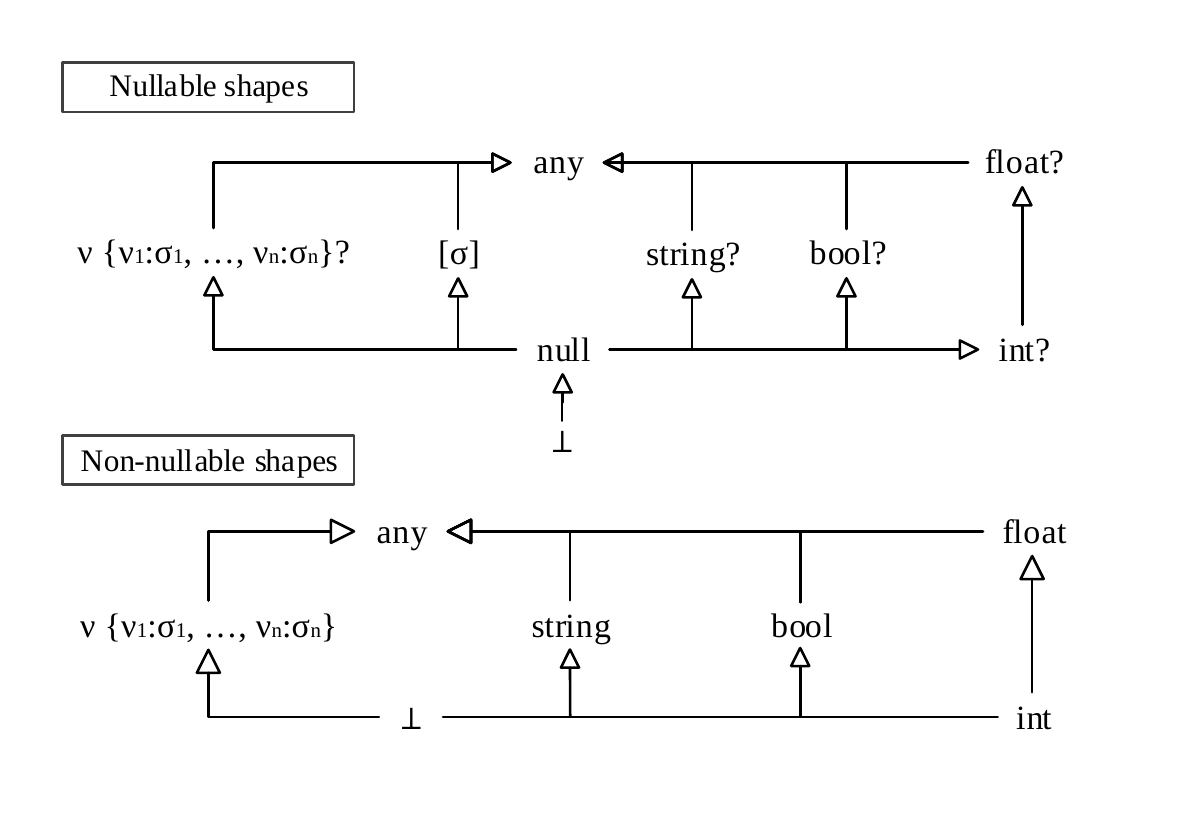} 
\end{center}
\vspace{-0.5em}
\caption{Important aspects of the preferred shape relation}
\label{fig:subtyping-diagram}
\end{figure}


\subsection{Preferred shape relation}
\label{sec:inference-subtyping}

Figure~\ref{fig:subtyping-diagram} provides an intuition about the preference between
shapes. The lower part shows non-nullable shapes (with records and primitives) and the upper part
shows nullable shapes with \kvd{null}, collections and nullable shapes. In the diagram, we
abbreviate $\kvd{nullable}\langl\sigma\rangl$ as $\sigma?$ and we omit links between the two parts;
a shape $\hat{\sigma}$ is preferred over $\kvd{nullable}\langl\hat{\sigma}\rangl$.

\begin{definition}
For ground $\sigma_1$, $\sigma_2$ (i.e. without $\rho_i$ variables), we write $\sigma_1 \sqsubseteq \sigma_2$ to denote that $\sigma_1$ is preferred over $\sigma_2$.
The shape preference relation is defined as a transitive reflexive closure of the following rules:

\noindent
\begin{align}
  \label{eq:sub-prim}
  \ident{int}\,&\sqsubseteq\,\ident{float}&\\[-0.2em]
  \label{eq:sub-null}
  \kvd{null} &\sqsubseteq \sigma  &(\textnormal{for}~\sigma \neq \hat{\sigma})  \\[-0.2em]
  \label{eq:sub-opt}
  \hat{\sigma} &\sqsubseteq \kvd{nullable}\langl\hat{\sigma}\rangl  &(\textnormal{for all}~\hat{\sigma})\\[-0.2em]
  \label{eq:sub-opt-cov}
  \kvd{nullable}\langl\hat{\sigma_1}\rangl &\sqsubseteq
    \kvd{nullable}\langl\hat{\sigma_2}\rangl  &(\textnormal{if}~\hat{\sigma_1} \sqsubseteq \hat{\sigma_2})\\[-0.2em]
  \label{eq:sub-col}
  [\sigma_1] &\sqsubseteq [\sigma_2]  &(\textnormal{if}~\sigma_1 \sqsubseteq \sigma_2) \\[-0.2em]
  \label{eq:sub-bot}
  \bot &\sqsubseteq \sigma  &(\textnormal{for all}~\sigma)\\[-0.2em]
  \label{eq:sub-var-top}
  \sigma &\sqsubseteq \kvd{any}
\end{align}
\vspace{-2em}

\noindent
\begin{align}
\label{eq:sub-record1}
\begin{array}{l}
 \nu~\{ \nu_1\!:\!\sigma_1, .., \nu_n\!:\!\sigma_n \} \sqsubseteq \\
 \quad\nu~\{ \nu_1\!:\!\sigma_1', .., \nu_n\!:\!\sigma_n' \}
\end{array} \qquad ~~~~(\textnormal{if}~\sigma_i \sqsubseteq \sigma_i')
\end{align}

\vspace{-1.5em}

\begin{align}
\label{eq:sub-record2}
\begin{array}{l}
 \nu~\{ \nu_1\!:\!\sigma_1, .., \nu_n\!:\!\sigma_n \} \sqsubseteq \\
 \quad\nu~\{ \nu_1\!:\!\sigma_1, .., .., \nu_m\!:\!\sigma_m \}
\end{array} \quad~~\, (\textnormal{when}~m \leq n)
\end{align}

\end{definition}


\begin{figure*}[t]
\noindent
\begin{equation*}
\begin{array}{rcll}
\ident{csh}(\sigma, \sigma) &=& \sigma & (\emph{eq})\\
\ident{csh}([\sigma_1], [\sigma_2]) &=& [\ident{csh}(\sigma_1, \sigma_2)] & (\emph{list}) \\
\ident{csh}(\bot, \sigma) = \ident{csh}(\sigma, \bot) &=& \sigma & (\emph{bot})\\
\ident{csh}(\kvd{null}, \sigma) = \ident{csh}(\sigma, \kvd{null}) &=& \addopt{\sigma}& (\emph{null}) \\
\ident{csh}(\kvd{any}, \sigma) = \ident{csh}(\sigma, \kvd{any}) &=& \kvd{any} & (\emph{top})\\
\ident{csh}(\kvd{float}, \kvd{int}) = \ident{csh}(\kvd{int}, \kvd{float}) &=& \kvd{float} & (\emph{num})\\
\ident{csh}(\sigma_2, \kvd{nullable}\langl \hat{\sigma_1} \rangl) = \ident{csh}(\kvd{nullable}\langl \hat{\sigma_1} \rangl, \sigma_2) &=& \addopt{\ident{csh}(\hat{\sigma_1}, \sigma_2)} & (\emph{opt})\\
\ident{csh}(\nu \; \{ \nu_1 \!:\! \sigma_1,  \ldots \;, \nu_n \!:\! \sigma_n \},
  \nu \; \{ \nu_1 \!:\! \sigma'_1, \; \ldots \;, \nu_n \!:\! \sigma'_n \}) &=&
  \nu \; \{ \nu_1 \!:\! \ident{csh}(\sigma_1, \sigma'_1), \; \ldots \; , \nu_n \!:\! \ident{csh}(\sigma_n, \sigma'_n) \} & (\emph{recd}) \\
\ident{csh}(\sigma_1, \sigma_2) &=& \kvd{any} \quad(\textnormal{when}~\sigma_1\neq\nu\;\{ \ldots \}~\textnormal{or}~\sigma_2\neq\nu\;\{ \ldots \} ) & (\emph{any})\\
\end{array}
\end{equation*}

\begin{equation*}
\quad\qquad
\begin{array}{rcll}
 \addopt{\hat{\sigma}} &\narrow{=}& \kvd{nullable}\langl\hat{\sigma}\rangl &(\textnormal{non-nullable shapes})\\
 \addopt{\sigma} &\narrow{=}& \sigma &(\textnormal{otherwise})
\end{array}
\qquad\qquad
\begin{array}{rcll}
 \dropopt{\kvd{nullable}\langl\hat{\sigma}\rangl} &\narrow{=}& \hat{\sigma} &(\textnormal{nullable shape})\\
 \dropopt{\sigma} &\narrow{=}& \sigma &(\textnormal{otherwise})
\end{array}
\end{equation*}

\caption{The rules that define the common preferred shape function}
\label{fig:subtyping-cst}
\end{figure*}


\noindent
Here is a summary of the key aspects of the definition:
\begin{itemize}
\item Numeric shape with smaller range is preferred (\ref{eq:sub-prim}) and we choose 32-bit
\ident{int} over \ident{float} when possible.

\item The \kvd{null} shape is preferred over all nullable shapes (\ref{eq:sub-null}), i.e.
  all shapes excluding non-nullable shapes $\hat{\sigma}$. Any non-nullable shape is preferred
  over its nullable version (\ref{eq:sub-opt})

\item Nullable shapes and collections are
  covariant (\ref{eq:sub-opt-cov}, \ref{eq:sub-col}).

\item There is a bottom shape (\ref{eq:sub-bot}) and \kvd{any} behaves as the top shape, because
  any shape $\sigma$ is preferred over \kvd{any} (\ref{eq:sub-var-top}).

\item The record shapes are covariant (\ref{eq:sub-record1}) and preferred record can have
  additional fields (\ref{eq:sub-record2}).
\end{itemize}

\noindent


\subsection{Common preferred shape relation}
\label{sec:inference-commonsuper}

Given two ground shapes, the \emph{common preferred shape} is the least upper bound of the
shape with respect to the preferred shape relation. The least upper bound prefers records,
which is important for usability as discussed earlier (\S\ref{sec:providers-xml}).
\begin{definition}
A \emph{common preferred shape} of two ground shapes $\sigma_1$ and $\sigma_2$ is a shape
$\ident{csh}(\sigma_1, \sigma_2)$ obtained according to Figure~\ref{fig:subtyping-cst}.
The rules are matched from top to bottom.
\end{definition}

\noindent
The fact that the rules of \ident{csh} are matched from top to bottom resolves the ambiguity
between certain rules. Most importantly (\emph{any}) is used only as the last resort.

When finding a common shape of two records (\emph{recd}) we find common preferred shapes of their
respective fields. We can find a common shape of two different numbers (\emph{num}); for two collections,
we combine their elements (\emph{list}). When one shape is nullable (\emph{opt}), we find
the common non-nullable shape and ensure the result is nullable using $\addopt{-}$, which
is also applied when one of the shapes is \kvd{null} (\emph{null}).

When defined, \ident{csh} finds the unique least upper bound of the partially ordered
set of ground shapes (Lemma~\ref{thm:lub}).

\begin{lemma}[Least upper bound]
\label{thm:lub}
For ground $\sigma_1$ and $\sigma_2$, if $\ident{csh}(\sigma_1, \sigma_2) \vdash \sigma$ then $\sigma$ is a least upper bound by $\sqsupseteq$.
\end{lemma}
\begin{proof}
By induction over the structure of the shapes $\sigma_1,\sigma_2$. Note that \ident{csh} only infers the top shape
\kvd{any} when on of the shapes is the top shape (\emph{top}) or when there is no other option
(\emph{any}); a nullable shape is introduced in $\addopt{-}$ only when no non-nullable shape can
be used (\emph{null}), (\emph{opt}).
\end{proof}

\subsection{Inferring shapes from samples}
\label{sec:formal-inferval}

We now specify how we obtain the shape from data. As clarified later (\S\ref{sec:impl-parsing}),
we represent JSON, XML and CSV documents using the same first-order \emph{data} value:

\noindent
\begin{equation*}
\begin{array}{lcl}
  \\[-0.5em]
 d &\narrow{=}& i \lsep f \lsep s \lsep \kvd{true} \lsep \kvd{false} \lsep \kvd{null} \\[0.1em]
   &\narrow{|}& [d_1; \ldots; d_n] \lsep \nu~\{ \nu_1 \mapsto d_1, \ldots, \nu_n \mapsto d_n \}
\end{array}
\end{equation*}

~

\noindent
The definition includes primitive values ($i$ for integers, $f$ for floats
and $s$ for strings) and \kvd{null}. A collection is written as a
list of values in square brackets. A record starts with a name $\nu$, followed by a
sequence of field assignments $\nu_i \mapsto d_i$.

Figure~\ref{fig:shape-inference} defines a mapping $\semalt{d_1,\ldots,d_n}$ which turns a
collection of sample data $d_1, \ldots, d_n$ into a shape $\sigma$. Before applying $S$, we assume
each record in each $d_i$ is marked with a fresh row inference variable $\rho_i$.
We then choose a ground, minimal substitution $\theta$ for row variables. Because $\rho_i$ variables
represent potentially missing fields, the $\addopt{-}$ operator from Figure~\ref{fig:subtyping-cst}
is applied to all types in the vector.

This is sufficient to equate the record field labels and satisfy the pre-conditions in rule
(\emph{recd}) when multiple record shapes are combined. The \ident{csh} function is not defined for
two records with mis-matching fields, however, the fields can always be made to match, through
a substitution for row variables. In practice, $\theta$ is found via row
variable unification \cite{rows-remy}. We omit the details here. No $\rho_i$ variables remain after
inference as the substitution chosen is ground.

Primitive  values are mapped to their corresponding shapes.
When inferring a shape from multiple samples, we use the common preferred shape relation to find a
common shape for all values (starting with $\bot$). This operation is used
when calling a type provider with multiple samples and also when inferring the shape of collection
values.


\begin{figure}[!h]
\begin{equation*}
\begin{array}{rclcrclcrcl}
 \semalt{i} &\narrow{=}& \ident{int} && \semalt{\kvd{null}}  &\narrow{=}& \kvd{null} && \semalt{\kvd{true}} &\narrow{=}& \ident{bool}\\
 \semalt{f} &\narrow{=}& \ident{float} && \semalt{s} &\narrow{=}& \ident{string} && \semalt{\kvd{false}}  &\narrow{=}& \ident{bool}\\
\end{array}
\end{equation*}
\noindent
\vspace{-0.5em}
\begin{equation*}
\begin{array}{l}
 \semalt{[d_1; \ldots; d_n]} = [\semalt{d_1, \ldots, d_n}]
 \\[0.5em]
 \semalt{\nu~\{ \nu_1 \mapsto d_1, \ldots, \nu_n \mapsto d_n \}_{\rho_i}} =\\[0.1em]
 \qquad\nu~\{ \nu_1:\semalt{d_1}, \ldots, \nu_n :\semalt{d_n}, \addopt{\theta(\rho_i)} \}
 \\[0.5em]
 \semalt{d_1, \ldots, d_n} = \sigma_n \quad\textnormal{where}\\[0.1em]
 \qquad\sigma_0 = \bot,~ \forall i\in \{ 1.. n \}.~ \sigma_{i-1} \triangledown \semalt{d_i} \vdash \sigma_i \\[0.6em]
 \textnormal{Choose minimal}\ \theta\ \textnormal{by ordering}\ \sqsubseteq\ \textnormal{lifted over substitutions} \\[0.1em]
\end{array}
\end{equation*}
\caption{Shape inference from sample data}
\label{fig:shape-inference}
\vspace{-1em}
\end{figure}


\begin{figure*}
  \noindent
  \begin{equation*}
  \begin{array}{rclcl}
   \tytag &\narrow{=}& \ident{collection}  &\narrow{\lsep}& \ident{number} \\
          &\narrow{\lsep}& \ident{nullable} &\narrow{\lsep}& \ident{string}  \\
          &\narrow{\lsep}& \nu ~\lsep \ident{any} &\narrow{\lsep}& \ident{bool}
  \end{array}
  \quad\;\;\;
  \begin{array}{rcl}
   \tytagof(\ident{string}) &\narrow{=}& \ident{string}\\
   \tytagof(\ident{bool}) &\narrow{=}& \ident{bool}\\
   \tytagof(\ident{int}) &\narrow{=}& \ident{number}\\
   \tytagof(\ident{float}) &\narrow{=}& \ident{number}\\
  \end{array}
  \;\;
  \begin{array}{rcl}
   \tytagof(\kvd{any}\langl\sigma_1, \ldots, \sigma_n\rangl) &\narrow{=}& \ident{any}\\
   \tytagof(\nu\; \{ \nu_1 : \sigma_1, \; \ldots \; , \nu_n : \sigma_n \}) &\narrow{=}& \nu \\
   \tytagof(\kvd{nullable}\langl\hat{\sigma}\rangl) &\narrow{=}& \ident{nullable}\\
   \tytagof([\sigma]) &\narrow{=}& \ident{collection}\\
  \end{array}
  \end{equation*}

  \vspace{-0.5em}

  \begin{equation*}
  \qquad\qquad
  \begin{array}{ll}
  \begin{array}{l}
    \ident{csh}(\kvd{any}\langl \sigma_1, \ldots, \sigma_k,  \ldots, \sigma_n\rangl,
      \kvd{any}\langl \sigma'_1, \ldots, \sigma'_k, \ldots, \sigma'_m\rangl) =\\
      \qquad \kvd{any}\langl \ident{csh}(\sigma_1, \sigma_1'), \ldots, \ident{csh}(\sigma_k, \sigma_k'),
        \sigma_{k+1}, \ldots, \sigma_{n}, \sigma'_{k+1}, \ldots, \sigma'_{m}\rangl \\[0.5em]
      \qquad \textnormal{For $i,j$ such that }(\tytagof(\sigma_i) = \tytagof(\sigma'_j)) \Leftrightarrow (i = j) \wedge (i \leq k)
  \end{array}\quad~
  &  (\emph{top-merge})
  \\[2.5em]
  \begin{array}{l}
    \ident{csh}(\sigma, \kvd{any}\langl\sigma_1, \ldots, \sigma_n\rangl) =
    \ident{csh}(\kvd{any}\langl\sigma_1, \ldots, \sigma_n\rangl, \sigma) = \\
    \qquad \kvd{any}\langl\sigma_1, \ldots, \dropopt{\ident{csh}(\sigma, \sigma_i)}, \ldots, \sigma_n\rangl\\[0.5em]
    \qquad \textnormal{For $i$ such that }\tytagof(\sigma_i) = \tytagof(\dropopt{\sigma})
  \end{array}
  &  (\emph{top-incl})
  \\[2.5em]
  ~\;\ident{csh}(\sigma, \kvd{any}\langl\sigma_1, \ldots, \sigma_n\rangl) =
   \kvd{any}\langl\sigma_1, \ldots, \sigma_n, \dropopt{\sigma}\rangl
   & (\emph{top-add})\\[0.5em]
  ~\;\ident{csh}(\sigma_1,\sigma_2) =
  \kvd{any}\langle\dropopt{\sigma_1}, \dropopt{\sigma_2}\rangle
  & (\emph{top-any})
 \end{array}
\end{equation*}

\caption{Extending the common preferred shape relation for labelled top shapes}
\label{fig:subtyping-cst-var}
\end{figure*}


\subsection{Adding labelled top shapes}
\label{sec:inference-vars}

When analyzing the structure of shapes, it suffices to consider a single top shape \kvd{any}.
The type providers need more information to provide typed access to the possible
alternative shapes of data, such as XML nodes.

We extend the core model (sufficient for the discussion of relative safety) with \emph{labelled top
shapes} defined as:
\begin{equation*}
\sigma = \ldots \lsep \kvd{any}\langl \sigma_1, \ldots, \sigma_n\rangl
\end{equation*}
The shapes $\sigma_1, \ldots, \sigma_n$ represent statically known shapes that appear in the
sample and that we expose in the provided type. As discussed earlier (\S\ref{sec:providers-xml})
this is important when reading external \emph{open world} data. The labels do not affect the
preferred shape relation and $\kvd{any}\langl \sigma_1, \ldots, \sigma_n\rangl$ should still be
seen as the top shape, regardless of the labels\footnote{An alternative would be to add unions of
shapes, but doing so in a way that is compatible with the open-world assumption breaks the
existence of unique lower bound of the preferred shape relation.}.

The common preferred shape function is extended to find a labelled top shape that best represents
the sample. The new rules for \kvd{any} appear in Figure~\ref{fig:subtyping-cst-var}.
We define shape \emph{tags} to identify shapes that have a common preferred shape
which is not the top shape. We use it to limit the number of labels and avoid nesting by grouping
shapes by the shape tag. Rather than inferring $\kvd{any}\langl\ident{int}, \kvd{any}\langl\ident{bool}, \ident{float}\rangl\rangl$,
our algorithm joins $\ident{int}$ and $\ident{float}$ and produces
$\kvd{any}\langl\ident{float},\ident{bool}\rangl$.

When combining two top shapes (\emph{top-merge}), we group the annotations by their tags.
When combining a top with another shape, the labels may or may not already contain a case with
the tag of the other shape. If they do, the two shapes are combined (\emph{top-incl}), otherwise
a new case is added (\emph{top-add}). Finally, (\emph{top-all}) replaces earlier (\emph{any})
and combines two distinct non-top shapes. As top shapes implicitly permit \kvd{null} values,
we make the labels non-nullable using $\dropopt{-}$.

The revised algorithm still finds a shape which is the least upper bound. This means that
labelled top shape is only inferred when there is no other alternative.

Stating properties of the labels requires refinements to the \emph{preferred shape} relation.
We leave the details to future work, but we
note that the algorithm infers the best labels in the sense that there are labels that enable
typed access to every possible value in the sample, but not more. The same is the case for nullable
fields of records.

%
%

\section{Formalizing type providers}
\label{sec:formal}

This section presents the formal model of F\# Data integration. To represent the programming
language that hosts the type provider, we introduce the Foo calculus, a subset of F\# with objects
and properties, extended with operations for working with weakly typed structured data along the
lines of the F\# Data runtime. Finally, we describe how
type providers turn inferred shapes into Foo classes (\S\ref{sec:formal-tp}).


\begin{figure}[!h]
\vspace{-0.1em}
\noindent
\begin{equation*}
\begin{array}{rcl}
 \tau &\narrow{=}& \ident{int} \lsep \ident{float} \lsep \ident{bool} \lsep \ident{string} \lsep C \lsep \ident{Data} \\[0.0em]
      &\narrow{|}& \tau_1 \rightarrow \tau_2 \lsep \ident{list}\langl\tau\rangl \lsep \ident{option}\langl\tau\rangl
\\[0.6em]
 L &\narrow{=}& \kvd{type}~C(\overline{x:\tau}) = \overline{M} \\[0.0em]
 M &\narrow{=}& \kvd{member}~N:\tau= e
\\[0.6em]
 v &\narrow{=}& d \lsep \ident{None} \lsep \ident{Some}(v) \lsep \kvd{new}~C(\overline{v}) \lsep v_1::v_2 \\[0.0em]
 e &\narrow{=}& d \lsep op \lsep e_1~e_2 \lsep \lambda x.e \lsep e.N \lsep \kvd{new}~C(\overline{e}) \\
   &\narrow{|}& \ident{None} \lsep\kvd{match}~e~\kvd{with}~\ident{Some}(x) \rightarrow e_1 \,|\, \ident{None} \rightarrow e_2 \\
   &\narrow{|}& \ident{Some}(e) \lsep e_1 = e_2 \lsep {\kvd{if}~e_1~\kvd{then}~e_2~\kvd{else}~e_3} \lsep \kvd{nil} \\
   &\narrow{|}& e_1 :: e_2 \lsep \kvd{match}~e~\kvd{with}~x_1::x_2 \rightarrow e_1 \,|\, \kvd{nil} \rightarrow e_2
\\[0.6em]
op &\narrow{=}& \ident{convFloat}(\sigma, e) \lsep \ident{convPrim}(\sigma, e) \\
   &\narrow{|}& \ident{convField}(\nu_1, \nu_2, e, e) \lsep \ident{convNull}(e_1, e_2) \\
   &\narrow{|}& \ident{convElements}(e_1, e_2) \,\lsep \ident{hasShape}(\sigma, e)
\end{array}
\end{equation*}

\caption{The syntax of the Foo calculus}
\label{fig:foo-syntax}
\vspace{-0.5em}
\end{figure}


\begin{figure*}
\noindent
\begin{center}
\textbf{Part I.} Reduction rules for conversion functions
\end{center}

\vspace{0.5em}
\noindent
\begin{equation*}
\begin{array}{l}
\ident{hasShape}(\nu~ \{ \nu_1 \!:\! \sigma_1, \ldots, \nu_n \!:\! \sigma_n \}, \nu'~\{ \nu'_1\mapsto d_1, \ldots, \nu'_m\mapsto d_m \}) \reduce (\nu = \nu') ~\wedge \\
  \quad (~ ((\nu_1 = \nu'_1) \wedge \ident{hasShape}(\sigma_1, d_1)) \vee\ldots\vee ((\nu_1 = \nu'_m) \wedge \ident{hasShape}(\sigma_1, d_m)) \vee \ldots \vee\\
  \quad ~\; ((\nu_n = \nu'_1) \wedge \ident{hasShape}(\sigma_n, d_1)) \vee\ldots\vee ((\nu_n = \nu'_m) \wedge \ident{hasShape}(\sigma_n, d_m))~)
\\[0.5em]
\ident{hasShape}([\sigma], [d_1; \ldots; d_n]) \reduce \ident{hasShape}(\sigma, d_1)\wedge\ldots\wedge\ident{hasShape}(\sigma, d_n) \\
\ident{hasShape}([\sigma], \kvd{null}) \reduce \kvd{true}
\end{array}
\quad
\begin{array}{l}
\\[-2.5em]
\ident{convFloat}(\ident{float}, i) \reduce f~(f=i) \\
\ident{convFloat}(\ident{float}, f) \reduce f \\[0.75em]
\ident{convNull}(\kvd{null}, e) \reduce \ident{None} \\
\ident{convNull}(d, e) \reduce \ident{Some}(e~d)
\end{array}
\end{equation*}
\vspace{-0.75em}
\begin{equation*}
\hspace{-1.1em}
\begin{array}{l}
\ident{hasShape}(\ident{string}, s) \reduce \kvd{true} \\
\ident{hasShape}(\ident{int}, i) \reduce \kvd{true}\\
\ident{hasShape}(\ident{bool}, d) \reduce \kvd{true} \quad(\textnormal{when}~d\in{\kvd{true},\kvd{false}} )\\
\ident{hasShape}(\ident{float}, d) \reduce \kvd{true} \quad(\textnormal{when}~d=i ~\textnormal{or}~ d=f) \\
\ident{hasShape}(\_, \_) \reduce \kvd{false} \\
\end{array}
\;\;
\begin{array}{l}
\ident{convPrim}(\sigma, d) \reduce d\quad(\sigma,d\in\{ (\ident{int},i),(\ident{string},s),(\ident{bool}, b) \})
\\[0.25em]
\ident{convField}(\nu,\nu_i, \nu~\{\ldots, \nu_i=d_i, \ldots\}, e) \reduce e~d_i\\
\ident{convField}(\nu,\nu', \nu~\{\ldots, \nu_i=d_i, \ldots\}, e) \reduce e~\kvd{null}\quad(\nexists i.\nu_i=\nu' )
\\[0.25em]
\ident{convElements}([d_1; \ldots; d_n], e) \reduce e~d_1 :: \ldots :: e~d_n :: \kvd{nil}  \\
\ident{convElements}(\kvd{null}) \reduce \kvd{nil}
\end{array}
\end{equation*}

\noindent
\begin{center}
\textbf{Part II.} Reduction rules for the rest of the Foo calculus
\end{center}

\vspace{0.5em}
\noindent
\begin{equation*}
\begin{array}{rl}
 \textnormal{\footnotesize{(member)}}&
 \hspace{-0.6em}
 \inference
 { \kvd{type}~C(\overline{x:\tau})= \kvd{member}~N_i : \tau_i = e_i \ldots \in L }
 { L, (\kvd{new}~C(\overline{v})).N_i \reduce e_i[\overline{x} \leftarrow \overline{v}] }\\
 \\[-0.2em]
 \textnormal{\footnotesize{(cond1)}}&
 \hspace{-0.4em}
 \kvd{if}~\kvd{true}~\kvd{then}~e_1~\kvd{else}~e_2 ~\reduce~ e_1 \\
 \\[-0.2em]
 \textnormal{\footnotesize{(cond2)}}&
 \hspace{-0.4em}
 \kvd{if}~\kvd{false}~\kvd{then}~e_1~\kvd{else}~e_2 ~\reduce~ e_2 \\
 \\[-0.3em]
 \textnormal{\footnotesize{(eq1)}}&
 v=v'\reduce\kvd{true} \qquad (\textnormal{when}~v = v')\\
 \\[-0.3em]
 \textnormal{\footnotesize{(eq2)}}&
 v=v'\reduce\kvd{false} \qquad (\textnormal{when}~v \neq v')\\
 \\[-0.3em]
 \textnormal{\footnotesize{(fun)}}&
 \hspace{-0.4em}
 (\lambda x.e)~v ~\reduce~ e[x\leftarrow v] \\
\end{array}
\quad
\begin{array}{rl}
 \textnormal{\footnotesize{(match1)}}&
 \hspace{-1em}
 \begin{array}{l}
  \kvd{match}~\ident{None}~\kvd{with} \\
  \ident{Some}(x) \rightarrow e_1 \,|\, \ident{None} \rightarrow e_2
 \end{array} \hspace{-0.5em} ~\reduce~ e_2 \\
 \\[-0.2em]
 \textnormal{\footnotesize{(match2)}}&
 \hspace{-1em}
 \begin{array}{l}
    \kvd{match}~\ident{Some}(v)~\kvd{with} \\
    \ident{Some}(x) \rightarrow e_1 \,|\, \ident{None} \rightarrow e_2
 \end{array} \hspace{-0.5em} ~\reduce~ e_1[x\leftarrow v]\\
 \\[-0.2em]
 \textnormal{\footnotesize{(match3)}}&
 \hspace{-1em}
 \begin{array}{l}
  \kvd{match}~\kvd{nil}~\kvd{with} \\[0em]
  x_1::x_2 \rightarrow e_1 \,|\, \kvd{nil} \rightarrow e_2
 \end{array} \hspace{-0.5em} ~\reduce~ e_2\\
 \\[-0.2em]
 \textnormal{\footnotesize{(match4)}}&
 \hspace{-1em}
 \begin{array}{l}
  \kvd{match}~v_1::v_2~\kvd{with} \\[0em]
  x_1::x_2 \rightarrow e_1 \,|\, \kvd{nil} \rightarrow e_2
 \end{array} \hspace{-0.5em} ~\reduce~ e_1[\overline{x}\leftarrow\overline{v}] \\
\\[-0.3em]
\textnormal{\footnotesize{(ctx)}}&
\hspace{-0.4em}
E[e] \reduce E[e'] \qquad\qquad(\textnormal{when}~e \reduce e')\\
\end{array}
\end{equation*}

\caption{Foo -- Reduction rules for the Foo calculus and dynamic data operations}
\label{fig:ff-reduction}
\end{figure*}


Type providers for structured data map the ``dirty'' world of weakly typed
structured data into a ``nice'' world of strong types. To model this, the Foo calculus
does not have \kvd{null} values and data values $d$ are never directly exposed.
Furthermore Foo is simply typed: despite using class types and object notation
for notational convenience, it has no subtyping.


\subsection{The Foo calculus}
\label{sec:formal-ff}

The syntax of the calculus is shown in Figure~\ref{fig:foo-syntax}.
The type \ident{Data} is the type of structural data $d$. A class definition $L$ consists of
a single constructor and zero or more parameter-less members. The declaration implicitly closes
over the constructor parameters.
Values $v$ include previously defined data $d$; expressions
$e$ include class construction, member access, usual functional constructs (functions,
lists, options) and conditionals. The $op$ constructs are discussed next.


\paragraph{Dynamic data operations.}

The Foo programs can only work with \ident{Data} values using certain primitive operations. Those are modelled
by the $op$ primitives. In F\# Data, those are internal and users never access them directly.

The behaviour of the dynamic data operations is defined by the reduction rules in
Figure~\ref{fig:ff-reduction} (Part I). The typing is shown in Figure~\ref{fig:ff-typecheck} and is
discussed later. The \ident{hasShape} function represents a runtime shape test.
It checks whether a \ident{Data} value $d$ (Section~\ref{sec:formal-inferval}) passed as the second
argument has a shape specified by the first argument. For records, we have to check that for each
field $\nu_1, \ldots, \nu_n$ in the record, the actual record value has a field of the same name
with a matching shape. The last line defines a ``catch all'' pattern, which returns \kvd{false}
for all remaining cases. We treat $e_1 \vee e_2$ and $e_1 \wedge e_2$ as a syntactic sugar for
\kvd{if}\,.\,.\,\kvd{then}\,.\,.\,\kvd{else} so the result of the reduction is just a Foo
expression.

The remaining operations convert data values into values of less preferred shape.
The \ident{convPrim} and \ident{convFloat} operations take the required shape and a data value.
When the data does not match the required type, they do not reduce. For example,
$\ident{convPrim}(\kvd{bool}, \num{42})$ represents a stuck state, but $\ident{convFloat}(\kvd{float}, \num{42})$
turns an integer \num{42} into a floating-point numerical value \num{42.0}.

The \ident{convNull}, \ident{convElements} and \ident{convField} operations take an additional
parameter $e$ which represents a function to be used in order to convert a contained value (non-null
optional value, list elements or field value); \ident{convNull} turns \kvd{null} data value into
\ident{None} and \ident{convElements} turns a data collection $[d_1, \ldots, d_n]$ into a Foo list
$v_1 :: \ldots :: v_n :: \kvd{nil}$ and a \kvd{null} value into an empty list.


\begin{figure*}
\noindent
\begin{equation*}
\begin{array}{c}
\inference
  {~}
  {\hspace{-0.5em}L; \Gamma \vdash d : \ident{Data}\hspace{-0.5em}}
~~
\inference
  {~}
  {\hspace{-0.5em}L; \Gamma \vdash i : \ident{int}\hspace{-0.5em}}
~~
\inference
  {~}
  {\hspace{-0.5em}L; \Gamma \vdash f : \ident{float}\hspace{-0.5em}}
~~
\inference
  {L; \Gamma, x : \tau_1 \vdash e : \tau_2}
  {L; \Gamma \vdash \lambda x.e : \tau_2}
~~
\inference
  {L; \Gamma \vdash e_2 : \tau_1 &
   L; \Gamma \vdash e_1 : \tau_1 \rightarrow \tau_2}
  {L; \Gamma \vdash e_1~e_2 : \tau_2}
\\[1.75em]
\inference
  {L; \Gamma \vdash e : \ident{Data}}
  {L; \Gamma \vdash \ident{hasShape}(\sigma, e) : \ident{bool}}
\quad
\inference
  { L; \Gamma \vdash e : \ident{Data} & \tau \in \{ \ident{int}, \ident{float} \} }
  { L; \Gamma \vdash \ident{convFloat}(\sigma, e) : \ident{float} }
\quad
\inference
  { L; \Gamma \vdash e_1 : \ident{Data} & L; \Gamma \vdash e_2 : \ident{Data} \rightarrow \tau }
  { L; \Gamma \vdash \ident{convNull}(e_1, e_2) : \ident{option}\langle\tau\rangle }
\\[1.75em]
\inference
  { L; \Gamma \vdash e : \ident{Data} \\
    \ident{prim}\in\{ \ident{int}, \ident{string}, \ident{bool} \} }
  { L; \Gamma \vdash \ident{convPrim}(\ident{prim}, e) : \ident{prim} }
\quad
\inference
  { L; \Gamma \vdash e_1 : \ident{Data} \\ L; \Gamma \vdash e_2 : \ident{Data} \rightarrow \tau }
  { L; \Gamma \vdash \ident{convElements}(e_1, e_2) : \ident{list}\langle\tau\rangle }
\quad
\inference
  { L; \Gamma \vdash e_1 : \ident{Data} \\ L; \Gamma \vdash e_2 : \ident{Data} \rightarrow \tau }
  { L; \Gamma \vdash \ident{convField}(\nu, \nu', e_1, e_2) : \tau }
\\[1.75em]
\inference
  {L; \Gamma \vdash e : C \\ \kvd{type}~C(\overline{x:\tau}) = ..\;\kvd{member}~N_i : \tau_i = e_i\;.. \in L}
  {L; \Gamma \vdash e.N_i:\tau_i}
\qquad
\inference
  {L; \Gamma \vdash e_i : \tau_i & \kvd{type}~C(x_1:\tau_1, \ldots, x_n:\tau_n) = \ldots \in L}
  {L; \Gamma \vdash \kvd{new}~C(e_1, \ldots, e_n):C}
\end{array}
\end{equation*}

\caption{Foo -- Fragment of type checking}
\label{fig:ff-typecheck}
\end{figure*}


\paragraph{Reduction.}
The reduction relation is of the form $L, e \reduce e'$. We omit class declarations
$L$ where implied by the context and write $e \reduce^{*} e'$ for
the reflexive, transitive closure of $\reduce$.

Figure~\ref{fig:ff-reduction} (Part II) shows the reduction rules.  The (\emph{member}) rule reduces
a member access using a class definition in the assumption. The (\emph{ctx}) rule models the eager
evaluation of F\# and performs a reduction inside a sub-expression specified by an evaluation
context $E$:
\begin{equation*}
\begin{array}{rcl}
 E &\narrow{=}& v::E \lsep v~E \lsep E.N \lsep \kvd{new}~C(\overline{v}, E, \overline{e})\\[0.1em]
   &\narrow{|}&  \kvd{if}~E~\kvd{then}~e_1~\kvd{else}~e_2  \lsep E = e \lsep v = E \\[0.1em]
   &\narrow{|}& \ident{Some}(E) \lsep op(\overline{v}, E, \overline{e})\\[0.1em]
   &\narrow{|}& \kvd{match}~E~\kvd{with}~\ident{Some}(x) \rightarrow e_1 \,|\, \ident{None} \rightarrow e_2 \\[0.1em]
   &\narrow{|}& \kvd{match}~E~\kvd{with}~x_1 :: x_2 \rightarrow e_1 \,|\, \kvd{nil} \rightarrow e_2
\end{array}
\end{equation*}

\noindent
The evaluation proceeds from left to right as denoted by $\overline{v}, E, \overline{e}$ in
constructor and dynamic data operation arguments or $v::E$ in list initialization.
We write $e[\overline{x} \leftarrow \overline{v}]$ for the result of replacing variables $\overline{x}$ by
values $\overline{v}$ in an expression. The remaining six rules
give standard reductions.


\paragraph{Type checking.}
Well-typed Foo programs reduce to a value in a finite number of steps or get stuck due to an
error condition. The stuck states can only be due to the dynamic data operations (e.g. an attempt
to convert \kvd{null} value to a number $\ident{convFloat}(\ident{float},\kvd{null})$). The relative safety (Theorem~\ref{thm:safety})
characterizes the additional conditions on input data under which Foo programs do not get stuck.

Typing rules in Figure~\ref{fig:ff-typecheck} are written using a judgement
$L; \Gamma \vdash e : \tau$ where the context also contains a set of class declarations $L$.
The fragment demonstrates the differences and similarities with Featherweight Java \cite{fwjava} and
typing rules for the dynamic data operations $op$:
\begin{itemize}
\item[--] All data values $d$ have the type \ident{Data}, but primitive data values (Booleans,
  strings, integers and floats) can be implicitly converted to Foo values and so they also have a
  primitive type as illustrated by the rule for $i$ and $f$.

\item[--] For non-primitive data values (including \kvd{null}, data collections and records),
  \ident{Data} is the only type.

\item[--] Operations $op$ accept \ident{Data} as one of the arguments and produce a non-\ident{Data}
  Foo type. Some of them require a function specifying the conversion for nested values.

\item[--] Rules for checking class construction and member access are similar to corresponding
  rules of Featherweight Java.
\end{itemize}
An important part of Featherweight Java that is omitted here is the checking of type declarations
(ensuring the members are well-typed). We consider only classes generated by our type providers
and those are well-typed by construction.


\subsection{Type providers}
\label{sec:formal-tp}

So far, we defined the type inference algorithm which produces a shape $\sigma$ from one
or more sample documents (\S\ref{sec:inference}) and we defined a simplified model of evaluation
of F\# (\S\ref{sec:formal-ff}) and F\# Data runtime (\S\ref{sec:formal-tp}). In this section, we
define how the type providers work, linking the two parts.

All F\# Data type providers take (one or more) sample documents, infer a common preferred shape $\sigma$
and then use it to generate F\# types that are exposed to the programmer.\footnote{The actual
implementation provides \emph{erased types} as described in \cite{fsharp-typeprov}. Here, we treat
the code as actually generated. This is an acceptable simplification, because F\# Data type providers
do not rely on laziness or erasure of type provision.}

\paragraph{Type provider mapping.}
A type provider produces an F\# type $\tau$ together with a Foo expression and a collection of
class definitions. We express it using the following mapping:
\begin{equation*}
\sem{\sigma} = (\tau, e, L) \qquad (\textnormal{where}~L,\emptyset \vdash e:\ident{Data}\rightarrow\tau)
\end{equation*}


\noindent
The mapping $\sem{\sigma}$ takes an inferred shape $\sigma$. It returns an F\# type $\tau$ and
a function that turns the input data (value of type \ident{Data}) into a Foo value of type $\tau$.
The type provider also generates class definitions that may be used by $e$.

Figure~\ref{fig:tp-generation} defines $\sem{-}$. Primitive types are handled by a single rule that
inserts an appropriate conversion function; \ident{convPrim} just checks that the shape matches
and \ident{convFloat} converts numbers to a floating-point.


\begin{figure*}
\begin{multicols}{2}

\noindent
\begin{equation*}
\begin{array}{l}
 \sem{\sigma_p} = \tau_p,\lambda x. op(\sigma_p, x),\emptyset\quad\textnormal{where} \\[0.6em]
 \quad\sigma_p, \tau_p, op\in  \{~ (\ident{bool}, \ident{bool}, \ident{convPrim})\\
 \hspace{2.9em} (\ident{int}, \ident{int}, \ident{convPrim}), (\ident{float},\ident{float},\ident{convFloat}),\\
 \hspace{2.9em} (\ident{string},\ident{string},\ident{convPrim}) ~\}
\end{array}
\end{equation*}
\vspace{-2em}

\begin{equation*}
\begin{array}{l}
 \sem{\,\nu~ \{ \nu_1 : \sigma_1, \ldots, \nu_n : \sigma_n \}\,} = \\[0.1em]
 \quad C, \lambda x. \kvd{new}~C(x), L_1\cup\ldots\cup L_n\cup\{ L \}\quad\textnormal{where}\\[0.6em]
 \qquad \;\;C~\textnormal{is a fresh class name} \\
 \qquad \;\,\,L = \kvd{type}~C(x_1\!:\!\ident{Data})=M_1 \ldots M_n \\
 \qquad M_i = \kvd{member}~\nu_i:\tau_i=\ident{convField}(\nu, \nu_i, x_1, e_i),\\
 \qquad \tau_i, e_i, L_i = \sem{\sigma_i}
\end{array}
\end{equation*}
\vspace{-2em}

\begin{equation*}
\begin{array}{l}
 \sem{\,[\sigma]\,} = \ident{list}\langl\tau\rangl, \lambda x . \ident{convElements}(x, e'), L \;\;\textnormal{where}\\[0.4em]
 \qquad \tau, e', L = \sem{\hat{\sigma}}
\end{array}
\end{equation*}

\noindent
\begin{equation*}
\begin{array}{l}
 \sem{\,\kvd{any}\langl\sigma_1, \ldots, \sigma_n\rangl\,} = \\
 \quad C, \lambda x.\kvd{new}~C(x), L_1\cup\ldots\cup L_n\cup\{ L \}\;\;\textnormal{where}\\[0.6em]
 \qquad \;\;C~\textnormal{is a fresh class name} \\
 \qquad \;\,\,L = \kvd{type}~C(x:\ident{Data})~=~M_1\ldots M_n \\
 \qquad M_i = \kvd{member}~\nu_i:\ident{option}\langl\tau_i\rangl=\\
 \hspace{5.8em}  \kvd{if}~\ident{hasShape}(\sigma_i, x)~\kvd{then}~ \ident{Some}(e_i~x)~\kvd{else}~\ident{None} \\[0.1em]
 \qquad \tau_i, e_i, L_i = \sem{\sigma_i}_e,\;\; \nu_i=\tytagof{(\sigma_i)}
\end{array}
\end{equation*}
\vspace{-2em}

\begin{equation*}
\hspace{-1.3em}
\begin{array}{l}
 \sem{\kvd{nullable}\langl\hat{\sigma}\rangl} = \\[0.2em]
 \qquad \ident{option}\langl\tau\rangl, \lambda x . \ident{convNull}(x,e), L\\[0.2em]
 \qquad \textnormal{where}~\tau, e, L = \sem{\hat{\sigma}}
\end{array}
\end{equation*}
\vspace{-2em}

\begin{equation*}
\hspace{-1.3em}
\begin{array}{l}
 \sem{\bot} = \sem{\kvd{null}} = C, \lambda x. \kvd{new}~C(x), \{ L \} \;\;\textnormal{where}\\[0.4em]
 \qquad C~\textnormal{is a fresh class name} \\[0.1em]
 \qquad L = \kvd{type}~C(v:\ident{Data})
\end{array}
\end{equation*}
\end{multicols}

\caption{Type provider -- generation of Foo types from inferred structural types}
\label{fig:tp-generation}
\vspace{-0.5em}
\end{figure*}


For records, we generate a class $C$ that takes a data value as a constructor parameter. For each
field, we generate a member with the same name as the field. The body of the member calls
\ident{convField} with a function obtained from $\sem{\sigma_i}$. This function turns the field
value (data of shape $\sigma_i$) into a Foo value of type $\tau_i$. The returned expression creates a new instance of
$C$ and the mapping returns the class $C$ together with all recursively generated classes. Note that
the class name $C$ is not directly accessed by the user and so we can use an arbitrary name, although the
actual implementation in F\# Data attempts to infer a reasonable name.\footnote{For example, in
\ident{\{\str{person}:\{\str{name}:\str{Tomas}\}\}}, the nested record will be named \ident{Person}
based on the name of the parent record field.}

A collection shape becomes a Foo $\ident{list}\langl\tau\rangl$. The returned expression calls \ident{convElements}
(which returns the empty list for data value \kvd{null}). The last parameter is the recursively obtained
conversion function for the shape of elements $\sigma$. The handling of the nullable shape is similar,
but uses \ident{convNull}.

As discussed earlier, labelled top shapes are also generated as Foo classes with properties. Given
$\kvd{any}\langl\sigma_1, \ldots, \sigma_n\rangl$, we get corresponding F\# types $\tau_i$ and generate
$n$ members of type $\ident{option}\langl \tau_i\rangl$. When the member is accessed, we need to perform
a runtime shape test using \ident{hasShape} to ensure that the value has the right shape (similarly to runtime
type conversions from the top type in languages like Java). If the shape matches, a \ident{Some} value is
returned. The shape inference algorithm also guarantees that there is only one case for each shape tag
(\S\ref{sec:inference-commonsuper}) and so we can use the tag for the name of the generated member.

\paragraph{Example 1.}
To illustrate how the mechanism works, we consider two examples. First, assume
that the inferred shape is a record
$\ident{Person}~\{~\ident{Age}\!:\!\kvd{option}\langl\ident{int}\rangl,~\ident{Name}\!:\!\ident{string}~ \}$.
The rules from Figure~\ref{fig:tp-generation} produce the \ident{Person} class shown below
with two members.

The body of the \ident{Age} member uses \ident{convField} as specified by the case for optional
record fields. The field shape is nullable and so \ident{convNull} is used in the continuation to
convert the value to \ident{None} if \ident{convField} produces a \kvd{null} data value and
\ident{hasShape} is used to ensure that the field has the correct shape. The \ident{Name} value should
be always available and should have the right shape so \ident{convPrim} appears directly in the
continuation. This is where the evaluation can get stuck if the field value was missing:
\begin{equation*}
\begin{array}{l}
 \kvd{type}~\ident{Person}(x_1:\ident{Data})~= \\[0.1em]
 \quad \kvd{member}~\ident{Age}~:~\ident{option}\langl\ident{int}\rangl~= \\[0.1em]
 \qquad \ident{convField}(\ident{Person}, \ident{Age}, x_1, \lambda x_2 \rightarrow \\[0.1em]
 \qquad \quad \ident{convNull}(x_2, \lambda x_3\rightarrow\ident{convPrim}(\ident{int}, x_3))~) \\[0.1em]
 \quad \kvd{member}~\ident{Name}~:~\ident{string}~= \\[0.1em]
 \qquad \ident{convField}(\ident{Person},\ident{Name}, x_1, \lambda x_2 \rightarrow \\[0.1em]
 \qquad \quad \ident{convPrim}(\ident{string}, x_2)))
\end{array}
\end{equation*}
The function to create the Foo value \ident{Person} from a data value
is $\lambda x . \kvd{new}~\ident{Person}(x)$.

\paragraph{Example 2.} The second example illustrates the handling of collections and
labelled top types. Reusing \ident{Person} from the previous example, consider
$[\kvd{any}\langl\ident{Person}~\{ \ldots \},\ident{string}\rangl]$:
\begin{equation*}
\begin{array}{l}
 \kvd{type}~\ident{PersonOrString}(x:\ident{Data})~= \\[0.1em]
 \quad \kvd{member}~\ident{Person}~:~\ident{option}\langl\ident{Person}\rangl~= \\[0.1em]
 \qquad \kvd{if}~\ident{hasShape}(\ident{Person}~\{ \ldots \}, x)~\kvd{then}\\[0.1em]
 \qquad\quad \ident{Some}(\kvd{new}~\ident{Person}(x))~\kvd{else}~\ident{None} \\[0.1em]
 \quad \kvd{member}~\ident{String}~:~\ident{option}\langl\ident{string}\rangl~= \\[0.1em]
 \qquad \kvd{if}~\ident{hasShape}(\ident{string}, x)~\kvd{then}\\[0.1em]
 \qquad\quad \ident{Some}(\ident{convPrim}(\ident{string}, x))~\kvd{else}~\ident{None}
\end{array}
\end{equation*}
The type provider maps the collection of labelled top shapes to a type $\ident{list}\langl\ident{PersonOrString}\rangl$
and returns a function that parses a data value as follows:

\noindent
\begin{equation*}
\lambda x_1\rightarrow \ident{convElements}(x_1\lambda x_2\rightarrow\kvd{new}~\ident{PersonOrString}(x_2))
\end{equation*}
The \ident{PersonOrString} class contains one member for each of the labels. In the body, they
check that the input data value has the correct shape using \ident{hasShape}. This also implicitly
handles \kvd{null} by returning \kvd{false}. As discussed earlier, labelled top types provide easy
access to the known cases (\ident{string} or \ident{Person}), but they require a runtime shape check.

%
%

\section{Relative type safety}
\label{sec:safety}

Informally, the safety property for structural type providers states that, given representative sample
documents, any code that can be written using the provided types is guaranteed to work. We call this
\emph{relative safety}, because we cannot avoid \emph{all} errors. In particular, one can always
provide an input that has a different structure than any of the samples. In this case, it is expected
that the code will throw an exception in the implementation (or get stuck in our model).

More formally, given a set of sample documents, code using the provided type is guaranteed to work if
the inferred shape of the input is preferred with respect to the shape of any of the samples. Going back to
\S\ref{sec:inference-subtyping}, this means that:
\begin{itemize}
\item[--] Input can contain smaller numerical values (e.g., if a sample contains float, the input can contain an integer).
\item[--] Records in the input can have additional fields.
\item[--] Records in the input can have fewer fields than some of the records in the sample
  document, provided that the sample also contains records that do not have the field.
\item[--] When a labelled top type is inferred from the sample, the actual input can also contain any other value,
  which implements the open world assumption.
\end{itemize}
The following lemma states that the provided code (generated in Figure~\ref{fig:tp-generation})
works correctly on an input $d'$ that is a subshape of $d$. More formally, the provided
expression (with input $d'$) can be reduced to a value and, if it is a class,
all its members can also be reduced to values.

\begin{lemma}[Correctness of provided types]
\label{thm:tp-correctness}
Given sample data $d$ and an input data value $d'$ such that $\semalt{d'} \sqsubseteq \semalt{d}$
and provided type, expression and classes $\tau, e, L = \sem{\semalt{d}}$,
then $L, e~d' \reduce^{*} v$ and if $\tau$ is a class ($\tau=C$) then for all members $N_i$ of the
class $C$, it holds that $L, (e~d').N_i \reduce^{*} v$.
\end{lemma}
\begin{proof}
By induction over the structure of $\sem{-}$. For primitives, the conversion functions accept all subshapes.
For other cases, analyze the provided code to see that it can work on all subshapes (for example~\ident{convElements}
works on \kvd{null} values, \ident{convFloat} accepts an integer). Finally, for labelled top types,
the \ident{hasShape} operation is used to guaranteed the correct shape at runtime.
\end{proof}

\noindent
This shows that provided types are correct with respect to the preferred shape relation.
Our key theorem states that, for any input which is a subshape the inferred shape and
any expression $e$, a well-typed program that uses the provided types does not ``go wrong''.
Using standard syntactic type safety  \cite{syntactic}, we prove type preservation
(reduction does not change type) and progress (an expression can be reduced).

\begin{theorem}[Relative safety]
\label{thm:safety}
Assume $d_1, \ldots, d_n$ are samples, $\sigma=\semalt{d_1, \ldots, d_n}$ is an inferred
shape and $\tau,e,L = \sem{\sigma}$ are a type, expression and class definitions generated by a
type provider.

For all inputs $d'$ such that $\semalt{d'} \sqsubseteq \sigma$ and all expressions $e'$
(representing the user code) such that $e'$ does not contain any of the dynamic data operations $op$
and any \ident{Data} values as sub-expressions and $L; y\!:\!\tau \vdash e'\!:\!\tau'$, it is
the case that $L, e[y\leftarrow e'~d'] \reduce^{*} v$ for some value $v$ and
also $\emptyset; \vdash v : \tau'$.
\end{theorem}
\begin{proof}
We discuss the two parts of the proof separately as type preservation (Lemma~\ref{thm:rs-preservation})
and progress (Lemma~\ref{thm:rs-progress}).
\end{proof}

\begin{lemma}[Preservation]
\label{thm:rs-preservation}
Given the $\tau, e, L$ generated by a type provider as specified in
the assumptions of Theorem~\ref{thm:safety}, then if $L, \Gamma \vdash e : \tau$ and
$L, e \reduce^{*} e'$ then $\Gamma \vdash e' : \tau$.
\end{lemma}
\begin{proof}
By induction over $\reduce$. The cases for the ML subset of Foo
are standard. For (\emph{member}), we check that code generated by type providers
in Figure~\ref{fig:tp-generation} is well-typed.
\end{proof}

\noindent
The progress lemma states that evaluation of a well-typed program does not reach an undefined state.
This is not a problem for the Standard ML \cite{sml} subset and object-oriented subset \cite{fwjava} of the calculus. The
problematic part are the dynamic data operations (Figure~\ref{fig:ff-reduction}, Part I). Given a data
value (of type \ident{Data}), the reduction can get stuck if the value does not have a structure
required by a specific operation.

The Lemma~\ref{thm:tp-correctness} guarantees that this does not happen inside the provided type.
We carefully state that we only consider expressions $e'$ which
``[do] not contain primitive operations $op$ as sub-expressions''. This ensure that only
the code generated by a type provider works directly with data values.

\begin{lemma}[Progress]
\label{thm:rs-progress}
Given the assumptions and definitions from Theorem~\ref{thm:safety}, there exists $e''$ such that
$e'[y\leftarrow e~d'] \reduce e''$.
\end{lemma}
\begin{proof}
Proceed by induction over the typing derivation of $L; \emptyset \vdash e[y\leftarrow e'~d'] : \tau'$.
The cases for the ML subset are standard. For member access, we rely on Lemma~\ref{thm:tp-correctness}.
\end{proof}

%
%

\section{Practical experience}
\label{sec:impl}

The F\# Data library has been widely adopted by users and is one of the most downloaded
F\# libraries.\footnote{At the time of writing, the library has over 125,000 downloads on NuGet
(package repository), 1,844 commits and 44 contributors on GitHub.} A practical demonstration of
development using the library can be seen in an attached screencast and additional documentation can be
found at \url{http://fsharp.github.io/FSharp.Data}.

In this section, we discuss our experience with the safety guarantees provided by the
F\# Data type providers and other notable aspects of the implementation.

\subsection{Relative safety in practice}
\label{sec:safety-discuss}

The \emph{relative safety} property does not guarantee safety in the same way as traditional
closed-world type safety, but it reflects the reality of programming with external data that is
becoming increasingly important \cite{age-of-web}. Type providers increase the safety of this kind of
programming.

\paragraph{Representative samples.}
When choosing a representative sample document, the user does not need to provide a sample
that represents all possible inputs. They merely need to provide a sample that is representative
with respect to data they intend to access. This makes the task of choosing a representative
sample easier.

\paragraph{Schema change.}
Type providers are invoked at compile-time. If the schema changes (so that inputs are no longer
related to the shape of the sample used at compile-time), the program can fail at runtime and
developers have to handle the exception. The same problem happens when using weakly-typed code
with explicit failure cases.

F\# Data can help discover such errors earlier. Our first example (\S\ref{sec:introduction})
points the JSON type provider at a sample using a live URL. This has the advantage that a
re-compilation fails when the sample changes, which is an indication that the program needs to be
updated to reflect the change.

\paragraph{Richer data sources.}
In general, XML, CSV and JSON data sources without an explicit schema will necessarily require
techniques akin to those we have shown. However, some data sources provide an explicit schema with
versioning support. For those, a type provider that adapts automatically could be written,
but we leave this for future work.


\subsection{Parsing structured data}
\label{sec:impl-parsing}

In our formalization, we treat XML, JSON and CSV uniformly as \emph{data values}. With the addition of
names for records (for XML nodes), the definition of structural values is rich enough to capture all
three formats.\footnote{The same mechanism has later been used by the HTML type provider
(\url{http://fsharp.github.io/FSharp.Data/HtmlProvider.html}), which provides similarly easy
access to data in HTML tables and lists.} However, parsing real-world data poses a number of practical issues.

\paragraph{Reading CSV data.}
When reading CSV data, we read each row as an unnamed record and return a collection of rows.
One difference between JSON and CSV is that in CSV, the literals have no data types and so
we also need to infer the shape of primitive values. For example:
{\small{
\begin{verbatim}
  Ozone, Temp, Date,       Autofilled
  41,    67,   2012-05-01, 0
  36.3,  72,   2012-05-02, 1
  12.1,  74,   3 kveten,   0
  17.5,  #N/A, 2012-05-04, 0
\end{verbatim}
}}
\noindent
The value {\small\ttfamily \#N/A} is commonly used to represent missing values in CSV and is treated
as \kvd{null}. The \ident{Date} column uses mixed formats and is inferred as \ident{string}
(we support many date formats and ``May 3'' would be parsed as date). More interestingly,
we also infer \ident{Autofiled} as Boolean, because the sample contains only $0$ and $1$.
This is handled by adding a \ident{bit} shape which is preferred of both \ident{int} and \ident{bool}.

\paragraph{Reading XML documents.}
Mapping XML documents to structural values is more interesting. For each node, we
create a record. Attributes become record fields and the body becomes a field with a special
name. For example:
{\small{
\begin{verbatim}
  <root id="1">
    <item>Hello!</item>
  </root>
\end{verbatim}
}}
\noindent
This XML becomes a record \ident{root} with fields \ident{id} and $\bullet$ for the body.
The nested element contains only the $\bullet$ field with the inner text. As with CSV, we
infer shape of primitive values:
\begin{equation*}
\ident{root}~\{ \ident{id} \mapsto 1, \bullet \mapsto [ \ident{item}~\{ \bullet \mapsto \str{Hello!} \}] \}
\end{equation*}
The XML type provider also includes an option to use \emph{global inference}. In that case,
the inference from values (\S\ref{sec:formal-inferval}) unifies the shapes of \emph{all} records with the
same name. This is useful because, for example, in XHTML all {\small\ttfamily <table>} elements
will be treated as values of the same type.


\subsection{Providing idiomatic F\# types}
\label{sec:impl-naming}

In order to provide types that are easy to use and follow the F\# coding guidelines,
we perform a number of transformations on the provided types that simplify their structure
and use more idiomatic naming of fields. For example, the type provided for the XML document in
\S\ref{sec:impl-parsing} is:
\begin{equation*}
\begin{array}{l}
 \kvd{type}~\ident{Root}~=  \\
 \quad \kvd{member}~\ident{Id}~:~\ident{int} \\
 \quad \kvd{member}~\ident{Item}~:~\ident{string}
\end{array}
\end{equation*}
To obtain the type signature, we used the type provider as defined in Figure~\ref{fig:tp-generation}
and applied three additional transformations and simplifications:

\begin{itemize}
\item When a class $C$ contains a member $\bullet$, which is a class with further members, the
  nested members are lifted into the class $C$. For example, the above type \ident{Root}
  directly contains \ident{Item} rather than containing a member $\bullet$ returning a
  class with a member \ident{Item}.

\item Remaining members named $\bullet$ in the provided classes (typically of primitive
  types) are renamed to \ident{Value}.

\item Class members are renamed to follow \ident{PascalCase} naming convention, when a
  collision occurs, a number is appended to the end as in \ident{PascalCase2}. The provided
  implementation preforms the lookup using the original name.
\end{itemize}

Our current implementation also adds an additional member to each class that returns the
underlying JSON node (called \ident{JsonValue}) or XML element (called \ident{XElement}).
Those return the standard .NET or F\# representation of the value and can be used to dynamically
access data not exposed by the type providers, such as textual values inside mixed-content XML elements.


\subsection{Heterogeneous collections}
\label{sec:impl-hetero}

When introducing type providers (\S\ref{sec:providers-sum}), we mentioned how F\# Data
handles heterogeneous collections. This allows us to avoid inferring
labelled top shapes in many common scenarios. In the earlier example, a sample collection
contains a record (with \strf{pages} field) and a nested collection with values.

Rather than storing a single shape for the collection elements as in $[\sigma]$, heterogeneous
collections store multiple possible element shapes together with their \emph{  inferred multiplicity}
(exactly one, zero or one, zero or more):
\begin{equation*}
\begin{array}{rcl}
 \psi &\narrow{=}& 1? \lsep 1 \lsep \ast \\
 \sigma &\narrow{=}& ~\ldots \lsep [\sigma_1, \psi_1 | \ldots | \sigma_n, \psi_n ]
\end{array}
\end{equation*}
We omit the details, but finding a preferred common shape of two heterogeneous
collections is analogous to the handling of labelled top types. We merge cases with the same tag (by finding
their common shape) and calculate their new shared multiplicity (for example, by turning
$1$ and $1?$ into $1?$).


\subsection{Predictability and stability}
\label{sec:impl-stable}

As discussed in \S\ref{sec:providers}, our inference algorithm is designed to be predictable
and stable. When a user writes a program using the provided type and then adds another sample
(e.g.~with more missing values), they should not need to restructure their program.
For this reason, we keep the algorithm simple. For example, we do not use probabilistic methods to
assess the similarity of record types, because a small change in the sample could cause a large change
in the provided types.

We leave a general theory of stability and predictability of type providers to future work, but
we formalize a brief observation in this section. Say we write a program using a provided type
that is based on a collection of samples. When a new sample is added, the program can be modified
to run as before with only small local changes.

For the purpose of this section, assume that the Foo calculus also contains an \ident{exn}
value representing a runtime exception that propagates in the usual way,
i.e.~$C[\ident{exn}]\reduce\ident{exn}$, and also a conversion function \ident{int} that
turns floating-point number into an integer.

\begin{remark}[Stability of inference]
Assume we have a set of samples $d_1, \ldots, d_n$, a provided type based on the samples
$\tau_1, e_1, L_1 = \sem{\semalt{d_1, \ldots, d_n}}$ and some user code $e$ written using
the provided type, such that $L_1; x:\tau_1\vdash e : \tau$.

Next, we add a new sample $d_{n+1}$ and consider a new provided type
$\tau_2, e_2, L_2 = \sem{\semalt{d_1, \ldots, d_n, d_{n+1}}}$.

Now there exists $e'$ such that $L_2; x:\tau_2\vdash e' : \tau$ and if
for some $d$ it is the case that $e[x\leftarrow e_1~d] \reduce v$ then
also $e'[x\leftarrow e_2~d] \reduce v$.

Such $e'$ is obtained by transforming sub-expressions of $e$ using one of the following
translation rules:
\begin{enumerate}
\item
$C[e]$ to $C[\kvd{match}~e~\kvd{with}~\ident{Some}(v) \rightarrow v~|~\ident{None} \rightarrow \ident{exn}]$
\item
$C[e]$ to $C[e.\ident{M}]$ where $\ident{M} = \tytagof(\sigma)$ for some $\sigma$
\item
$C[e]$ to $C[\ident{int}(e)]$
\end{enumerate}
\end{remark}

\begin{proof}
For each case in the type provision (Figure~\ref{fig:tp-generation}) an original shape $\sigma$
may be replaced by a less preferred shape $\sigma'$. The user code can always be transformed
to use the newly provided shape:

\begin{itemize}
\item[--] Primitive shapes can become nullable (1), \ident{int} can become \ident{float} (3)
  or become a part of a labelled top type (2).
\item[--] Record shape fields can change shape (recursively) and record may become a part
  of a labelled top type (2).
\item[--] For list and nullable shapes, the shape of the value may change (we apply the
  transformations recursively).
\item[--] For the $\kvd{any}$ shape, the original code will continue to work (none of the labels is ever removed).
\end{itemize}
\vspace{-1.5em}
\end{proof}
\noindent
Intuitively, the first transformation is needed when the new sample makes a type optional.
This happens when it contains a \kvd{null} value or a record that does not contain a field
that all previous samples have. The second transformation is needed when a shape $\sigma$
becomes $\kvd{any}\langle\sigma, \ldots\rangle$ and the third one is needed when $\ident{int}$
becomes \ident{float}.

This property also underlines a common way of handling errors when using F\# Data type providers.
When a program fails on some input, the input can be added as another sample. This makes some
fields optional and the code can be updated accordingly, using a variation of (i) that uses
an appropriate default value rather than throwing an exception.

%
%

\section{Related and future work}
\label{sec:related}

The F\# Data library connects two lines of research that have been previously disconnected. The first is
extending the type systems of programming languages to accommodate external data sources and the second
is inferring types for real-world data sources.

The type provider mechanism has been introduced in F\# \cite{fsharp-typeprov,fsharp-typeprov-ddfp},
added to Idris  \cite{idris-tp} and used in areas such as semantic web \cite{liteq}. The F\# Data
library has been developed as part of the early F\# type provider research, but previous
publications focused on the general mechanisms. This paper is novel in that it shows the
programming language theory behind a concrete type providers.

\paragraph{Extending the type systems.}
Several systems integrate external data into a programming language. Those include
XML \cite{xduce,xduce-ml} and databases \cite{links}. In both of these, the system requires
the user to explicitly define the schema (using the host language) or it has an ad-hoc extension
that reads the schema (\emph{e.g.}~from a database). LINQ \cite{linq} is more general, but relies
on code generation when importing the schema.

The work that is the most similar to F\# Data is the data integration in C$\omega$ \cite{comega-xs}.
It extends C\# language with types similar to our structural types
(including nullable types, choices with subtyping and heterogeneous collections with multiplicities).
However, C$\omega$ does not infer the types from samples and extends the type system of the host
language (rather than using a general purpose embedding mechanism).

In contrast, F\# Data type providers do not require any F\# language extensions. The simplicity
of the Foo calculus shows we have avoided placing strong requirements on the host language. We
provide nominal types based on the shapes, rather than adding an advanced
system of structural types into the host language.

\paragraph{Advanced type systems and meta-programming.}
A number of other advanced type system features could be used to tackle the problem discussed
in this paper. The Ur \cite{ur} language has a rich system for working with records;
meta-programming \cite{template-hask,th-camlp4} and multi-stage programming \cite{multi-stage}
could be used to generate code for the provided types; and gradual typing \cite{gradual,gradual-js}
can add typing to existing dynamic languages. As far as we are aware, none of these
systems have been used to provide the same level of integration with XML, CSV and JSON.

\paragraph{Typing real-world data.}
Recent work \cite{typing-json} infers a succinct type of large JSON datasets using MapReduce.
It fuses similar types based on similarity. This is more sophisticated than our technique, but it
makes formal specification of safety (Theorem~\ref{thm:safety}) difficult. Extending our
\emph{relative safety} to \emph{probabilistic safety} is an interesting future direction.

The PADS project \cite{pads-dsl,pads-ml} tackles a more general problem of handling \emph{any} data format.
The schema definitions in PADS are similar to our shapes. The structure inference for LearnPADS
\cite{pads-learn} infers the data format from a flat input stream. A PADS type provider could follow
many of the patterns we explore in this paper, but formally specifying the safety property would be
more challenging.

\section{Conclusions}
\label{sec:conclusions}

We explored the F\# Data type providers for XML, CSV and JSON. As most real-world data does not come
with an explicit schema, the library uses \emph{shape inference} that deduces a shape from a set of
samples. Our inference algorithm is based on a preferred shape relation. It prefers records to
encompass the open world assumption and support developer tooling. The inference algorithm is predictable, which is
important as developers need to understand how changing the samples affects the resulting types.

We explored the theory behind type providers. F\# Data is a prime example of
type providers, but our work demonstrates a more general point. The types generated by type
providers can depend on external input and so we can only guarantee \emph{relative safety},
which says that a program is safe only if the actual inputs satisfy additional conditions.

Type providers have been described before, but this paper is novel in that it explores the
properties of type providers that represent the ``types from data'' approach. Our experience suggests
that this significantly broadens the applicability of statically typed languages to real-world
problems that are often solved by error-prone weakly-typed techniques.

\acks
We thank to the F\# Data contributors on GitHub and other colleagues working
on type providers, including Jomo Fisher, Keith Battocchi and Kenji Takeda. We are
grateful to anonymous reviewers of the paper for their valuable feedback and to
David Walker for shepherding of the paper.

\bibliographystyle{abbrvnat}
\bibliography{paper}

\appendix

\section{OpenWeatherMap service response}
\label{sec:appendix-weather}

The introduction uses the \ident{JsonProvider} to access weather
information using the OpenWeatherMap service. After registering, you can access the service
using a URL \url{http://api.openweathermap.org/data/2.5/weather} with query string parameters
\strf{q} and \strf{APPID} representing the city name and application key. A sample response looks
as follows:

\vspace{-1em}
{\small\begin{verbatim}
{
  "coord": {
    "lon": 14.42,
    "lat": 50.09
  },
  "weather": [
    {
      "id": 802,
      "main": "Clouds",
      "description": "scattered clouds",
      "icon": "03d"
    }
  ],
  "base": "cmc stations",
  "main": {
    "temp": 5,
    "pressure": 1010,
    "humidity": 100,
    "temp_min": 5,
    "temp_max": 5
  },
  "wind": { "speed": 1.5, "deg": 150 },
  "clouds": { "all": 32 },
  "dt": 1460700000,
  "sys": {
    "type": 1,
    "id": 5889,
    "message": 0.0033,
    "country": "CZ",
    "sunrise": 1460693287,
    "sunset": 1460743037
  },
  "id": 3067696,
  "name": "Prague",
  "cod": 200
}
\end{verbatim}}

\end{document}